\documentclass[english, 11pt,a4paper]{article}

\usepackage{graphicx}
\usepackage{epstopdf}
\usepackage{amsmath}
\begin{document}


\font\msytw=msbm10 scaled\magstep1
\font\msytww=msbm7 scaled\magstep1
\font\msytwww=msbm5 scaled\magstep1
\font\cs=cmcsc10
\font\ottorm=cmr8

\let\a=\alpha \let\b=\beta  \let\g=\gamma  \let\d=\delta \let\e=\varepsilon
\let\z=\zeta  \let\h=\eta   \let\th=\theta \let\k=\kappa \let\l=\lambda
\let\m=\mu    \let\n=\nu    \let\x=\xi     \let\p=\pi    \let\r=\rho
\let\s=\sigma \let\t=\tau   \let\f=\varphi \let\ph=\varphi\let\c=\chi
\let\ps=\psi  \let\y=\upsilon \let\o=\omega\let\si=\varsigma
\let\G=\Gamma \let\D=\Delta  \let\Th=\Theta\let\L=\Lambda \let\X=\Xi
\let\P=\Pi    \let\Si=\Sigma \let\F=\Phi    \let\Ps=\Psi
\let\O=\Omega \let\Y=\Upsilon

\def\ins#1#2#3{\vbox to0pt{\kern-#2 \hbox{\kern#1 #3}\vss}\nointerlineskip}

\newdimen\xshift \newdimen\xwidth \newdimen\yshift


\def\vdd{{\vec d}}\def\vee{{\vec e}}\def\vkk{{\vec k}}\def\vii{{\vec i}}
\def\vmm{{\vec m}}\def\vnn{{\vec n}}\def\vpp{{\vec p}}\def\vqq{{\vec q}}
\def\vxxi{{\vec \xi}}\def\vrr{{\vec r}}\def\vtt{{\vec t}}
\def\vuu{{\vec u}}\def\bv{{\bf v}}
\def\vxx{{\vec x}}\def\vyy{{\vec y}}\def\vzz{{\vec z}}
\def\un{{\underline n}} \def\ux{{\underline x}} \def\uk{{\underline k}}
\def\xxx{{\underline\xx}}\def\vxx{{\vec x}} \def\vxxx{{\underline\vxx}}
\def\kkk{{\underline\kk}} \def\vkkk{{\underline\vkk}}
\def\bO{{\bf O}}\def\rr{{\bf r}}
\def\ss{{\underline \sigma}}\def\oo{{\underline \omega}}

\def\PPP{{\cal P}}\def\cE{{\cal E}}\def\cF{{\cal F}}

\def\cM{{\cal M}} \def\VV{{\cal V}}
\def\CC{{\cal C}}\def\FF{{\cal F}} \def\FFF{{\cal F}}

\def\HHH{{\cal H}}\def\WW{{\cal W}}
\def\TT{{\cal T}}\def\NN{{\cal N}} \def\BBB{{\cal B}}\def\III{{\cal I}}
\def\RR{{\cal R}}\def\LL{{\cal L}} \def\JJ{{\cal J}} \def\OO{{\cal O}}
\def\DD{{\cal D}}\def\AAA{{\cal A}}\def\GG{{\cal G}} \def\SS{{\cal S}}
\def\KK{{\cal K}}\def\UU{{\cal U}} \def\QQ{{\cal Q}} \def\XXX{{\cal X}}

\def\hh{{\bf h}} \def\HH{{\bf H}} \def\AA{{\bf A}} \def\bq{{\bf q}}
\def\BB{{\bf B}} \def\XX{{\bf X}} \def\PP{{\bf P}} \def\dd{{\bf d}} 
\def\bp{{\bf p}}
\def\vv{{\bf v}} \def\xx{{\bf x}} \def\yy{{\bf y}} \def\bk{{\bf k}}
\def\ba{{\bf a}}\def\bbb{{\bf b}}\def\bt{{\bf t}}\def\II{{\bf I}}
\def\ii{{\bf i}}\def\jj{{\bf j}}\def\bn{{\bf n}}\def\bS{{\bf S}}
\def\mm{{\bf m}}\def\Vn{{\bf n}}\def\uu{{\bf u}}\def\tt{{\bf t}}
\def\B{\hbox{\msytw B}}
\def\RRR{\hbox{\msytw R}} \def\rrrr{\hbox{\msytww R}}
\def\rrr{\hbox{\msytwww R}} \def\CCC{\hbox{\msytw C}}
\def\cccc{\hbox{\msytww C}} \def\ccc{\hbox{\msytwww C}}
\def\MMM{\hbox{\euftw M}}\font\euftw=eufm10 scaled\magstep1%
\def\NNN{\hbox{\msytw N}} \def\nnnn{\hbox{\msytww N}}
\def\nnn{\hbox{\msytwww N}} \def\ZZZ{\hbox{\msytw Z}}
\def\zzzz{\hbox{\msytww Z}} \def\zzz{\hbox{\msytwww Z}}
\def\SSS{{\bf S}}
\def\SSSS{\hbox{\euftwww S}}

\newcommand{\mR}{{\msytw R}}
\def\virg{\quad,\quad}


\def\\{\hfill\break}
\def\={:=}
\let\io=\infty
\let\0=\noindent\def\pagina{{\vfill\eject}}
\def\media#1{{\langle#1\rangle}}
\let\dpr=\partial
\def\sign{{\rm sign}}
\def\const{{\rm const}}
\def\tende#1{\,\vtop{\ialign{##\crcr\rightarrowfill\crcr\noalign{\kern-1pt
    \nointerlineskip} \hskip3.pt${\scriptstyle #1}$\hskip3.pt\crcr}}\,}
\def\otto{\,{\kern-1.truept\leftarrow\kern-5.truept\to\kern-1.truept}\,}
\def\defin{{\buildrel def\over=}}
\def\wt{\widetilde}
\def\wh{\widehat}
\def\to{\rightarrow}
\def\la{\left\langle}
\def\ra{\right\rangle}
\def\qed{\hfill\raise1pt\hbox{\vrule height5pt width5pt depth0pt}}
\def\Val{{\rm Val}}
\def\ul#1{{\underline#1}}
\def\lis{\overline}
\def\V#1{{\bf#1}}
\def\be{\begin{equation}}
\def\ee{\end{equation}}
\def\bea{\begin{eqnarray}}
\def\eea{\end{eqnarray}}
\def\bd{\begin{definition}}
\def\ed{\end{definition}}

\def\nn{\nonumber}
\def\pref#1{(\ref{#1})}
\def\ie{{\it i.e.}}
\def\cC{{\cal C}}
\def\lb{\label}
\def\eg{{\it e.g.}}
\def\sl{{\displaystyle{\not}}}
\def\Tr{\mathrm{Tr}}
\def\BBBB{\hbox{\msytw B}}
\def\bbb{\hbox{\msytww B}}
\def\TTT{\hbox{\msytw T}}
\def\bT{{\bf T}}
\def\mod{{\rm mod}}
\def\der{{\rm d}}
\def\bs{\backslash}
\newtheorem{corollary}{Corollary}[section]
\newtheorem{lemma}{Lemma}[section]
\newtheorem{example}{Example}[section]
\newtheorem{notation}{Notation}[section]
\newtheorem{remark}{Remark}[section]
\newtheorem{proof}{Proof}[section]
\newtheorem{definition}{Definition}[section]
\newtheorem{theorem}{Theorem}[section]
\newtheorem{proposition}{Proposition}[section]
\newtheorem{oss}{Remark}



%


\title{{\bf On the Sector Counting Lemma}}
\author{Zhituo Wang\\ Institute for Advanced Study in Mathematics, \\Harbin Institute of Technology\\
Email: wzht@hit.edu.cn}


\maketitle

\begin{abstract}
In this short note we prove a sector counting lemma for a class of Fermi surface on the plane which are $C^2$-differentiable and strictly convex. This result generalizes the one proved in \cite{FKT} for the class of $C^{2+r}$-differentiable, $r\ge3$, strictly convex and strongly asymmetric Fermi surfaces, and the one proved in \cite{FMRT} and \cite{BGM1}, for the class of $C^2$-differentiable, strictly convex and central symmetric Fermi surfaces. This new sector counting lemma can be used to construct interacting many-fermion models for the doped graphene, in which the Fermi surface is extended and quasi-symmetric.


\end{abstract}

\renewcommand{\thesection}{\arabic{section}}

\section{Introduction and Main results}
\subsection{The Fermi surface problem}
The Landau theory of the Fermi liquid \cite{landau1} is one the most important achievements in quantum many-body theory. It essentially states that, in a $d$-dimensional crystal, the excitations of an infinitely large collection of strongly interacting particles can be described as an equally large collection of weakly interacting quasi-particles, which carry the same quantum numbers as the original particles, and are characterized by a definite band structure $\e_0(\bk):=e(\bk)-\nu$ on $\RRR^d$, in which $e(\bk)$ is the dispersion relation and $\nu\in\RRR$ is the chemical potential. One important feature of a Fermi liquid is the existence of the Fermi surface (F.S.), which is defined as the zero set of the band structure: $\cF_0=\{\bk\in\RRR^d\ |\e_0(\bk)=0\}$. It is a compact hyper-surface in $\RRR^d$, across which the quasi-particle density function $n(\bk)$ is not continuous but has a jump. 

A major difficulty in the rigorous study of an interacting many-fermion system is that, interaction produces a deformation of the Fermi surface.
Consider a $d\ge2$-dimensional interacting many-fermions system at temperature $T>0$, defined by the grand-canonical Hamiltonian:
\bea\label{hamil}
&&H=\sum_\sigma\int\frac{d^d\bk}{(2\pi)^2}\e_0(\bk)a_{\bk,\s}^+a_{\bk,\s}+\sum_{\sigma,\tau}\frac{1}{2}\int\prod_{i=1}^4\frac{d^d\bk_i}{(2\pi)^d}(2\pi)^d\delta(\bk_1+\bk_2-\bk_3-\bk_4)\nn\\
&&\quad\quad\quad\times\hat u(\bk_1-\bk_3)a_{\bk_1,\s}^+a_{\bk_2,\tau}^+a_{\bk_4,\tau}a_{\bk_3,\s},
\eea 
in which $a^{\pm}$ are the fermionic creation and annihilation operators defined on $\RRR^d\times\{\uparrow,\downarrow\}$, with $\s,\tau\in\{\uparrow,\downarrow\}$ the spin indices. $\hat u$ is the Fourier transform of the two-body interaction potential. Let $k_0=2\pi T(n_0+\frac{1}{2})$, $n_0\in\ZZZ$, be the Fourier dual of the (imaginary) time variable $x_0\in[0,\frac{1}{T}]$. 
When $\hat u=0$, the free (non-interacting) propagator is $\hat S_0(k_0,\bk)=(ik_0-\e_0(\bk))^{-1}$. Obviously, the set of singularities of $\hat S_0(k_0,\bk)$ at $k_0=0$ is exactly $\cF_0$. When $\hat u\neq0$, the interacting propagator is given by $\hat S(k_0,\bk)=[ik_0-\e_0(\bk)-\Sigma((k_0,\bk),\e_0)]^{-1}$, in which $\Sigma((k_0,\bk),\e_0)$ is called the self-energy function, which is a highly non-trivial function of the band structure $\e_0$ and the interaction $\hat u$. For $||\hat u||$ small under a suitable norm $||\cdot||$, $\hat S(k_0,\bk)$ and $\Sigma(k_0,\bk)$ can be calculated by perturbation expansions.
The {\it interacting fermi surface} $\cF$ is defined as the set of singularities of $\hat S(k_0,\bk)$ at $k_0=0$:
\be
\cF_{\e_0}=\{\bk\in\RRR^d\ \vert\ \e_0(\bk)+\Sigma((0,\bk), \e_0)=0\},
\ee
which is in general different from $\cF_0$. 
This shift in the Fermi surface, also called the moving-Fermi surfaces problem, is a major difficulty in Quantum many-fermion problem and may cause divergence of many coefficients in the naive perturbation expansions.

This problem can be solved mainly in two approaches, one is to fix the interacting Fermi surface \cite{FST} such that it coincides with the non-interacting one $\cF_0$, by introducing
a suitable counter-term $\sum_\sigma\int\frac{d^d\bk}{(2\pi)^2}\delta \e_0(\bk)a_{\bk,\s}^+a_{\bk,\s}$ to the interaction potential. But this approach raises another difficulty, which is called {\it the inversion problem} (cf. eg. \cite{FST1}): given the band structure $\e(\bk):=\e_0(\bk)+\delta\e_0$, whether the counter-term $\delta\e_0$ can be uniquely decided, and how to determine $\e_0$ from $\e$. This problem has not been solved non-perturbatively. The other approach \cite{BGM1} is to use the renormalized interacting propagator $\hat S$, whose singular set defines the interacting Fermi surface, in the perturbation expansions. Either approach needs the renormalization group (RG) analysis. 

In order to perform the RG analysis, one needs to decompose the support of the propagators in $\RRR^d$ into a set of rectangles, called the {\it sectors}. The decomposition is performed in two successive steps. First of all, one decomposes the region in $\RRR^d$ close to the Fermi surface into shells surrounding the Fermi surface, with size of the shell depending on the temperature. But this is not enough to obtain the desired decaying behavior of the propagator, due to the mismatch of the volume form in momentum space and the position space \cite{FMRT}. One has to further decompose each shell into a set of sectors. This complicates the RG analysis, as in addition to the scaling indices labeling the shells, one needs to sum over properly the sector indices while taking into account the conservation of momentum. This is called the {\it sector counting problem}, which lies at the heart of Fermi liquid theory. A key step for proving the sector counting lemma is to estimate the flexibility of the constraint imposed by the conservation of momentum, namely, for any given $\bq\in\RRR^2$ and any Fermi surface $\cF$, the cardinality of the set $\{\bk_1,\bk_2\vert\ \bk_1,\bk_2\in\cF, \bk_1+\bk_2=\bq\}$. This is equivalent to consider the inverse image of the mapping $\Phi:\cF\times\cF\rightarrow\RRR^2$, $\Phi:(\bk_1,\bk_2)\mapsto\bq$, in which $\Phi$ is differentiable but not necessarily injective.  
The systematic way of estimating this number is cumulated into {\it the parallelogram lemmas}.

The sector counting lemmas for the convex and central symmetric Fermi surfaces have been proved 
by \cite{FMRT} \cite{DR1} and \cite{BGM1}. In particular, the authors of \cite{BGM1} have solved the inversion problem for the doped Hubbard model on the square lattice, following the second approach. But the sector counting lemma of \cite{BGM1} can't be applied to more general Fermi surfaces. The sector counting lemma for the strictly convex and strongly asymmetric Fermi Surfaces have been proved in \cite{FKT}. Based on this lemma and by introducing the counter-term, the authors constructed an interacting many-fermions model which exhibits Fermi liquid behaviors at zero temperature. But they can't provide a non-perturbative solution of the inversion problem. The sector counting lemma for Fermi surfaces that are not strictly convex but have flat edges have been proved in \cite{Riv} and \cite{RivWa, RivWb}.

\subsection{The main results}
In this paper we consider the $2$-d many-fermions models with two classes of Fermi surfaces: the {\it quasi-asymmetric} Fermi surface (cf. Definition \ref{fs3}) and the {\it quasi-symmetric} Fermi surfaces (Definition \ref{fs4}), both are required to be strongly convex and $C^2$ differentiable.
The first main result is Lemma \ref{main1}, in which we prove a parallelogram lemma for these Fermi surfaces. Since a quasi-symmetric Fermi surface is neither strongly asymmetric (see \cite{FKT}) not central symmetric (see \cite{FMRT}, \cite{DR1} and \cite{BGM1}), our result is an important generalization of these results. What's more, the parallelogram lemma is valid for any $C^2$-differentiable strongly convex Fermi surfaces, and for the class of concave Fermi surfaces such that at any point on the Fermi surface there exists only finite number of antipodal points (cf. Definition \ref{fs2}). Based on this result, we prove the sector counting lemma for a single scale 
(Theorem \ref{pro2}) and for multi-scales (Theorem \ref{pro3}).

Inspired by \cite{FKT}, we prove the main results by identifying the zero measure set ${\rm Gra}(a)=:\{(\bk,a(\bk))\vert\ \bk\in\cF\}$, the pairs of antipodal points, on which the mapping $\Phi$ fails to be injective. Since $\Phi$ is injective on the domain $\cF\times\cF\setminus {\rm Gra}(a)$, the implicit function theorem can be applied. Then we prove that the zero measure set doesn't change the result of {\it sector counting}. 

This paper is organized as follows. In section 2 we recall basic definitions, notations of the Fermi surfaces. In section 3 we prove the parallelogram lemma (Lemma \ref{main1}) for the $C^2$ strictly convex F.S. which are quasi-asymmetric or quasi-symmetric. This lemma is valid up to a measure zero set. In Section $4$ we prove that the measure zero set doesn't change the sector counting and present the sector counting lemmas for a single scale and for multi-scales of sectorizations.
We expect this result can be used to solve the inversion problem for the model considered by \cite{FKT1} without introducing counter-term, as well as constructing models with more general FS, like the honeycomb Hubbard model for the study of graphene. Some results proved in this paper have analogues in \cite{FKT} and their proofs are almost identical. So we simply omit these proofs and ask the interested readers to consult \cite{FKT} for details.



\section{Preliminary}
\subsection{Fermi surfaces and the sectors} 
Consider a many-fermions model on a $d=2$ lattice. 
\bd\label{fs1}
A Fermi surface is defined as the zero set of the dispersion relation. It is a closed curve in $\RRR^2$ which may contain several connected components, each of which is called a Fermi curve (F.C.). A Fermi curve (also denoted by $\cF$) is called $C^r$ differentiable, $r\ge1$, if the dispersion relation $\e(\bk)$ is a $C^r$-differentiable function in a neighborhood of $\bk$, for all $\bk\in\cF$; It is called strictly convex if its curvature is bounded away from zero.
\ed
Let us choose an orientation for the fermi curve $\cF$: For any $\bk\in\cF$, let $\bt_{\bk}$ be the unit tangent vector to $\cF$ at $\bk$ and $\bn_{\bk}$ the inward pointing unit normal vector to $\cF$ at $\bk$. There is a differentiable function $\phi_{\bk}:I_{0,a}\rightarrow\RRR$, where $I_{0,a}$ is an interval in $\RRR$ centered at $0$ with size $a$, such that $s\mapsto\bk+s\bt_{\bk}+\phi_{\bk}(s)\bn_{\bk}$ is an oriented parametrization of $\cF$ near $\bk$. By construction we have $\phi_{\bk}(0)=\phi'_{\bk}(0)=0$ and $\phi''_{\bk}(0)$ is the curvature of $\cF$ at $\bk$.

\begin{definition}\label{fs2}
Let $\bk\in\cF$ be any point on the Fermi curve, an antipodal point $a(\bk)$ of $\bk$, $a(\bk)\neq\bk$, is a point in $\cF$ such that the tangent vector $\bt_{a(\bk)}$ to $\cF$ at $a(\bk)$ is parallel or anti-parallel to $\bt_{\bk}$. 
\end{definition}
Remark that, since $\cF$ is strictly convex, each point $\bk\in\cF$ has a unique antipodal point.

\begin{definition}\label{fs3}
Let $\cF$ be a $C^r$-differentiable Fermi curve. It is called strongly asymmetric (also called $C^r$-strongly asymmetric) if there is $n_0\in\NNN$, $n_0\le r$, such that for each $\bk\in\cF$, there exists $p\le n_0$ such that 
\be
\phi^{(p)}_{\bk}(0)\neq \phi^{(p)}_{a(\bk)}(0).\label{assy}
\ee
A Fermi curve $\cF$ is called $C^r$-quasi asymmetric if it is $C^r$-strongly asymmetric at almost all points on the Fermi surface, up to a zero measure set. 

\ed

\bd\label{fs4}
A Fermi curve $\cF$ is called quasi-symmetric if it is strictly convex and the isometric group is the dihedral group $D_{2n+1}$, $n\ge1$, $n\in\NNN$, i.e. the symmetry group of regular polygons with $2n+1$ sides. It is called $C^r$-quasi-symmetric if it is $C^r$-differentiable.
\ed
See Figure \ref{tripod} for a quasi-symmetric Fermi curve with symmetry group $D_3$, which can be considered as the F.C. for the Graphene system with doping \cite{RivW}.
\begin{figure}[htp]
\centering
\includegraphics[width=.7\textwidth]{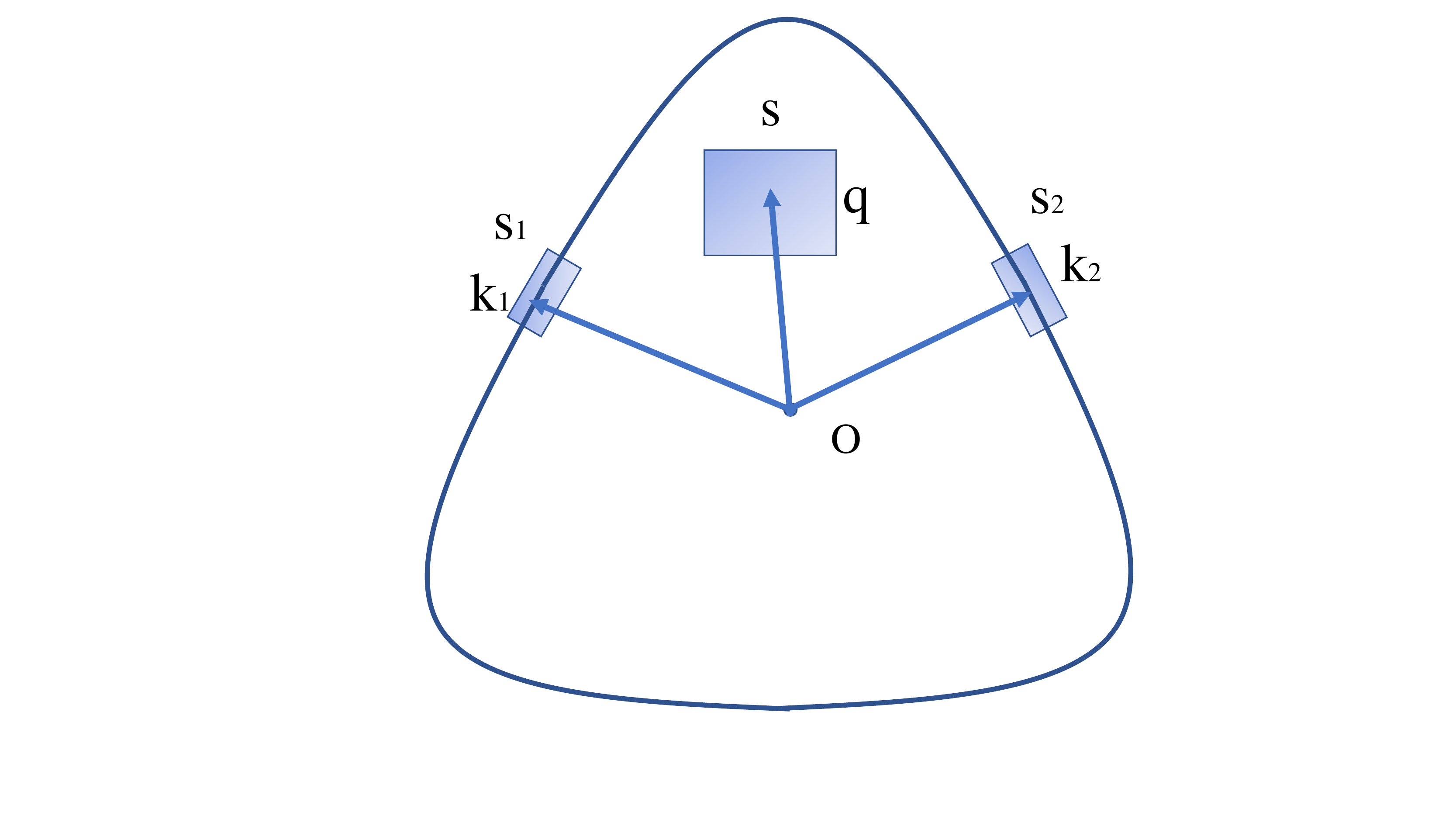}
\caption{\label{tripod}
A quasi-symmetric Fermi curve $\cF$ with $D_3$ symmetry and the parallelogram on $\cF$.}
\end{figure}
%

\bd[Shells and sectors]
By introducing suitable cutoff functions (eg. the Gevrey class of functions, cf. \cite{FMRT, RivW}), the support of the free propagator in the momentum space can be decomposed into shells labeled
by indices $j$, $j\ge 2$. Let 
$S^j=\{\bk\in\RRR^2\vert M^{-j}\le|\Omega(\bk)-\mu|\le M^{-j+1}\}$ be $j$-th shell in the momentum space.  Let $I$ be an interval on the Fermi curve $\cF$ and $\pi_\cF(\bk)$ is the orthogonal projection of $\bk$ on the Fermi curve. Then
\be
s=\{{\bk} \in S^j\vert\pi_\cF(\bk)\in I \}
\ee
is called a sector of length $|I|$ at scale $j$. Two different sectors $s$ and $s'$ are called neighbors if $s\cap s'\neq\emptyset$. 
\ed
\bd
A sectorization of length $l$ at scale $j$ around a Fermi curve $\cF$ is a set $\Sigma_\cF$ of sectors of length $l$ and at scale $j$ that obeys
\begin{itemize}
\item the set of sectors covers the Fermi curve;
\item each sector $s$ has precisely two neighbors in $\Sigma_\cF$.
\item if $s, s'\in\Sigma_\cF$ are neighbors then $\frac{1}{16}l\le|s\cap s'\cap \cF|\le \frac18 l$.
\end{itemize}
\ed
From the above definition we can easily find that there are at most $2\ell(\cF)/l$ sectors in $\Sigma_\cF$, where $\ell(\cF)$ is the length of $\cF$.

\begin{definition}\label{inverim}
Let $\cM$ be a sub-manifold of $\RRR^d$ and $\mu$ be the volume measure on $\BBB(\cM)$, the Borel $\sigma$-algebra over $\cM$. We say that a subset $A\subset \cM$ has $\mu$-measure zero if for every smooth chart $(U,\phi)$ for $\cM$, the subset $\phi(A\cap U)\subset \RRR^d$ has $m$-measure zero, where $m$ is the Lebesgue measure on $\RRR^d$.
Let $\Phi:\cM\rightarrow\RRR^d$ be a measurable mapping of class $C^1$ and is not necessarily injective. For any measurable $E\subset \cM$ and $y\in \RRR^n$, define the function $\#(E,y)=\#\{x\in E\vert \Phi(x)=y\}$.
In case $\Phi(x)=y$ for infinitely many $x\in E$, then we define $\#(E,y)=\infty$. 
\end{definition}

The Jacobian of the mapping $\Phi$ and the number $\#(E,y)$ is related by the following theorem:
\begin{theorem}\label{main2}
Let $\cM$ and $\Phi$ be defined as above. Let $J(x)=\Phi'(x)$ be the Jacobian at $x\in \cM$. Then for any measurable $E\subset \cM$, $\#(E,y)$ is a measurable function of $y\in \RRR^n$, and
\be
\int_E|J(x)|dx=\int_{\RRR^n}\#(E,y)dy.
\ee 
\end{theorem}
This theorem is already known. The interested reader could consult eg. \cite{munk}, pages 505-510, for a complete proof and we don't repeat it here. 

Feldman, Kn\"orrer and Trubowitz proved in \cite{FKT} the following parallelogram lemma concerning the strictly asymmetric Fermi curves.
\begin{lemma}[Lemma XX.7 of \cite{FKT}]\label{para}
Let $\cF\subset\RRR^2$ be a strongly asymmetric, strictly convex Fermi surface that is $C^{2+r}$ differentiable, $r\ge3$. Let $\BBB(\cF\times\cF)$ be the Borel sigma algebra over $\cF\times\cF$ and $\mu$ be the volume measure on $\BBB(\cF\times\cF)$. Define the measurable mapping $\Phi:\cF\times\cF\rightarrow\RRR^2$, $(\bk_1, \bk_2)\mapsto\bk_1+\bk_2$, which is not necessarily injective. Then there exists $X\subset\cF\times\cF$, for which $\mu(\cF\times\cF\setminus X)=0$, and a constant $C$, which depends on the geometry of $\cF$, such that for any measurable subset $E\subset X$ and any $\bq\in\RRR^2$, we have
\be
\#\{(E, \bq)\}\le C.
\ee
\end{lemma}
{\bf Proof}
The full measure set $X$ can be chosen as $X=\{(\bk_1,\bk_2)\vert \bk_1,\bk_2\in\cF,\bk_1\neq\bk_2\}$. The mapping $\Phi$ is injective on $X$ except at the points $\{(\bk,a(\bk))\}$, which form a measure zero set in $\cF\times\cF$. A detailed proof has been given in \cite{FKT}, so we don't repeat it here.
\qed
\section{The Parallelogram lemma}
As a generalization of Lemma \ref{para}, we can prove the following parallelogram lemmas.
\begin{lemma}\label{main1}
Let $\cF\subset\RRR^2$ be a strictly convex Fermi curve which is $C^2$ quasi-asymmetric or $C^2$ quasi-symmetric, let $\mu$ be the volume measure on $\BBB(\cF\times\cF)$. Define the mapping $\Phi:\cF\times\cF\rightarrow\RRR^2$, $\Phi:(\bk_1, \bk_2)\mapsto\bk_1+\bk_2$ which is not necessarily injective. Then there exists $X_1\subset\cF\times\cF$, with $\mu(\cF\times\cF\setminus X_1)=0$, and a constant $C$, which depends on the geometry of $\cF$, such that for any measurable subset $E\subset X_1$ and for almost every $\bq\in\RRR^2$, we have
\be\label{conmain1}
\#\{(E, \bq)\}\le C.
\ee
\end{lemma}
{\bf Proof}
Since $\cF$ is strictly convex, for any $\bk\in\cF$ there exists a unique antipodal point $a(\bk)$. Obviously $a:\cF\rightarrow\cF$, $a:\bk\mapsto a(\bk)$ is an isomorphism of $\cF$ and the graph of the mapping ${\rm Gra}(a):=\{(\bk, a(\bk))\vert \bk\in\cF, a(\bk)\in\cF\}$ is a $\mu$-measure zero set in $\cF\times\cF$. Since $\Phi$ is injective on $\cF\times\cF\setminus {\rm Gra}(a)$, we can identify the full measure set $X_1$ as
$X_1=X\setminus {\rm Gra}(a)$, in which $X=\{(\bk_1,\bk_2)\vert \bk_1,\bk_2\in\cF,\bk_1\neq\bk_2\}$.
By construction, no pair of antipodal points is contained in $X_1$. Then we only need to prove \eqref{conmain1} for pairs of vectors that are not antipodal of each other. This proof can be found in Lemma XX.7 in \cite{FKT}, and we don't repeat it here. 
\qed

\begin{remark}
It is important to notice that the conclusion of this lemma depends only on the convexity of the Fermi curve $\cF$ but not on the global symmetry of the F.C., i.e., if it is quasi-asymmetric or quasi symmetric, and this conclusion holds only up to a measure zero set. We will prove in the next section that the measure zero sets don't change the result of the sector counting. This can also be seen by the following simple argument. Let $\Sigma_\cF$ be any sectorization of $\cF$, whose elements are positive measure sets. Then any measure zero set in $\cF\times\cF$ must be contained in $\Sigma_\cF\times\Sigma_\cF$. So the existence of some measure zero sets doesn't change the result of counting sectors. 
\end{remark}

\begin{remark}
Remark that Lemma \ref{main1} can be generalized to the case of concave Fermi curves $\cF$ for which the antipodal points of any point $\bk\in\cF$ form a finite set.
\end{remark}

\section{The sector counting lemma}
In this part we shall consider the sector counting lemma for $2n$ sectors constrained by the conservation of momentum. As a first step, let us consider the case of $4$ sectors.
\subsection{Parallelogram lemma for two sectors}
Consider $4$ vectors $\bk_1,\cdots,\bk_4$ that belongs to the four sectors $s_1,\cdots,s_4$, respectively. We are interested in counting the cardinality of the configuration set of sectors $\{(s_1,\cdots,s_4)\vert\ \sum_i\bk_i=0,\ \bk_i\in s_i \}$ that are compatible with conservation of momentum. Let 
$\bk_3+\bk_4=-\bq$, this problem is equivalent to counting the number of decompositions of a subset $A\in\RRR^2$, to which $\bq$ belongs, into {\it sectors}. The result is also called the parallelogram lemma for two sectors.

\begin{definition}[Vector sum of sectors]
Let $\Sigma_\cF$ be a sectorization of a F.C. $\cF$, i.e. a set of sectors whose union form an $\e$-neighborhood $\cF(\e)$ of $\cF$. Let $s_1,s_2\in\Sigma_\cF$ be two sectors and $\bk_1\in s_1$, $\bk_2\in s_2$ be two vectors varying in the two sectors, respectively. Define a differentiable mapping $\Phi:\Sigma_\cF\times\Sigma_\cF\rightarrow\RRR^2$, $\Phi:\bk_1+\bk_2\mapsto\bq$. The image of the mapping $\Phi(s_1\times s_2)$, noted by
$\Phi(s_1\times s_2):=s_1+s_2$, is called the vector sum of the sectors $s_1$ and $s_2$. In the same way we can define the vector sum for any $n\ge2$ sectors.
\end{definition}
We are mainly interested in inverse problem of the vector sum of sectors, given a sectorization $\Sigma_\cF$ of a Fermi curve $\cF$ and a measurable subset $A\in\RRR^2$, the cardinality of the set $\{(s_1,s_2)\vert s_1,s_2\in\Sigma_\cF,\ s_1+s_2\subseteq A\}$. A general solution to this problem is called the {\it parallelogram lemma for sectors}, which is the simplest but most important example of the {\it sector counting lemma}. Instead of using techniques from planar differential geometry, Feldman, Kn\"orrer and Trubowitz proposed in \cite{FKT} a new method which reduces the counting problems to the problems of estimating volumes of sets in momentum space that are constrained by the conservation of momentum. Before proceeding, let us introduce the following definitions from Riemannian geometry.

\begin{definition}[\cite{gromov}]
Let $(\cM, \dd)$ be a $d$-dimensional Riemannian manifold with metric function $\dd$. Given $\epsilon>0$, a subset $\Gamma$ of $M$ is called $\epsilon$-separated if for any two different elements $\gamma, \gamma'\in\Gamma$, $\dd(\gamma,\gamma')\ge\epsilon$.
\end{definition}
Feldman, Kn\"orrer and Trubowitz proved in \cite{FKT} that:
\begin{lemma}[Lemma XX.4 in\cite{FKT}]\label{sec2a}
Let $(\cM, \dd)$ be a $d$-dimensional Riemannian manifold, $\Phi: \cM\rightarrow \RRR^d$ be a differentiable mapping. Let $B_r(x)=\{y\in M\vert d(x,y)<r\}$ be an open ball of radius $r$ around $x$, let $\mu$ be the volume measure on $\BBB(\cM)$, and
\be
V_{\cM,\epsilon}=\inf_{x\in M,\ 0<r\le\epsilon}\frac{\mu (B_{r/2}(x))}{r^2}.
\ee
Then for all $\epsilon_0>0$, $A\subset \RRR^d$ with $\epsilon$-neighborhood $A'(\epsilon)$, 
$0<\epsilon<\epsilon_0$, and all $\epsilon$-separated subsets $\Gamma\subset\cM$, we have:
\be
\#(\Phi^{-1}(A)\cap\Gamma)\le\frac{1}{\epsilon^n V_{\cM,\epsilon_0}}\mu(\Phi^{-1}(A'(\epsilon))).
\ee 
\end{lemma}

Taking $\cM=\cF\times\cF$ and $d=2$, using Lemma \ref{main1}, we can prove the following lemma, which is a key step for proving the sector counting lemma.
\begin{lemma}\label{sec2b}
Let $\cF$ be a $C^2$-differentiable strictly convex Fermi curve. Let $\omega_1$ and $\omega_2$ be two positive real numbers such that $0<\omega_1<\frac{1}{2}\omega_2$. For any $\bp\in\cF$, define the set
\bea
&&\tilde\cM=\{(\bk_1, \bk_2)\in\cF\times\cF\vert \min[d(\bk_1), \bk_2), d(a(\bk_1), \bk_2)]\ge\omega_1\nn\\
&&\quad {\rm and}\ \min[d(\bk_i), \bp), d(a(\bk_i), \bp)]\le\omega_2\ {\rm for}\ i=1,2 \},
\eea
and the mapping $\Phi:\cF\times\cF\rightarrow\RRR^2$, $\Phi(\bk_1, \bk_2)=\bk_1+\bk_2$.
Then there exists positive constants $K$ depending only on the geometry of $\cF$, such that for all measurable subset $A\subset\RRR^2$, 
\be
\mu(\Phi^{-1}(A)\cap \tilde\cM)\le\frac{const}{\omega_1}m(A),\label{eq1}
\ee
where $\mu$ is the volume measure on $\BBB(\cF\times\cF)$ and $m$ is the Lebesgue measure on $\RRR^2$.
\end{lemma}
{\bf Proof}
First of all, we can calculate explicitly the Jacobian $J(\bk_1, \bk_2)$, $(\bk_1, \bk_2)\in \tilde\cM$, for the mapping $\Phi$. Let $\theta(\bk_1,\bk_2)$ be the angle between the normal vectors to $\cF$ at $\bk_1$ and $\bk_2$, by simple calculations we find that 
\be J(\bk_1,\bk_2)=\sin\theta(\bk_1,\bk_2).\label{jcb}\ee 

From the definition of $\tilde\cM$ we find that:
\be
\vert \sin\theta(\bk_1,\bk_2)\vert\ge const \min[d(\bk_1), \bk_2), d(a(\bk_1), \bk_2)]\ge const\ \omega_1.
\ee
Now we consider the following two cases: 

(i), $\Phi^{-1}(A)\cap \tilde\cM=\emptyset$. Then \eqref{eq1} is obviously true. 

(ii), if $M_A:=\Phi^{-1}(A)\cap \tilde\cM\neq \emptyset$, then we have $\int_{A}\#(M_A,y)dy=\int_{M_A}|J(x)|dx$, by Theorem \ref{main2}. 
Let ${\rm Gra}(a)=\{(\bk,a(\bk))\vert\bk\in\cF\}\subset \cF\times\cF$ be the set of antipodal pairs and $M_A'=M_A\setminus {\rm Gra}(a)$, we have $\#(M_A,y)=\#(M'_A,y)$, for $a.e.\ y\in A$ (cf. \cite{munk}, Pages 505-510.).

Let $J_{min}$ and $J_{max}$ be the minimal and maximal value of the Jacobian 
$|J(x)|$, we have
\be
\int_{M'_A}|J|dx\ge J_{min} \mu(M'_A)\ge const\ \omega_1  \mu(M'_A)= const\ \omega_1  \mu(M_A),
\ee
where for the last equality we used the fact that $M_A=M'_A+{\rm Gra}(a)$ and ${\rm Gra}(a)$ is a zero measure set.
On the other hand we have
\be
\int_{A}\#(M_A,y)dy=\int_{A}\#(M'_A,y)dy \le \#(M'_A,y)_{max} m(A),
\ee
where $\#(M'_A,y)_{max}$ is the maximal number of pre-images of any $y\in A$. Since $\cF\setminus N_\cF$ is strongly asymmetric, we have $\#(M'_A,y)_{max}<n$, for a finite $n\in\NNN$.
So we have
\be
\#(M'_A,y)_{max} m(A)\ge \int_{A}\#(M_A,y)dy=\int_{M'_A}|J|dx\ge  const\ \omega_1  \mu(M_A).
\ee
and
\be
\mu(\Phi^{-1}(A)\cap \tilde\cM)=\mu(M_A)\le\frac{const}{\omega_1} m(A).
\ee
\qed

Combining the above two lemmas we can prove the following lemma, which is very similar to Lemma $XX.8$ of \cite{FKT}, except that the Fermi surface $\cF$ now has different geometric properties.
\begin{lemma}\label{main3}:
Let $0<\e<\omega_1/4$ and let $\Gamma$ be an $\e$-separated subset of $\cF$. Let $A$ be a rectangle in $\RRR^2$ having one pair of sides parallel to $\bn$ with length $L_1$ and a second pair of sides perpendicular to $\bn$ of length $L_2$. Then we have
\be
\#(\Phi^{-1}(A)\cap (\cF\times\cF)\cap(\Gamma\times\Gamma))\le\frac{const}{\omega_1\e^2}(L_1+\e\omega_2)(L_2+\e).
\ee
\end{lemma}
\begin{proof}
Using the results of Lemmas \ref{sec2a} we have
\be
\#(\Phi^{-1}(A)\cap (\cF\times\cF)\cap(\Gamma\times\Gamma))\le \frac{const}{\e^2}\mu(\Phi^{-1}(A'(\e))).
\ee
Then using Lemma \ref{sec2b}, we have $\mu(\Phi^{-1}(A'(\e)))\le\frac{const}{\omega_1}m(A'(\e))$
and the fact that $m(A'(\e))\le (L_1+\e\omega_2)(L_2+\e) $, the result follows.
\end{proof}

\subsection{The Sector counting lemma}

In the previous section we proved that any measure zero sets in $\cF\times\cF$ doesn't change the result of sector counting. This result can be easily generalized to the case of any $2n$ sectors, as the vector sum of $2n$ sectors can be reduced to the parallelogram lemma for $n$-sectors, which can be further reduced to the one for $2$ sectors inductively, by consider the vector sum of $n-1$ sectors as a $single$ sector. In this section we consider the sector counting lemma for general $2n$ sectors. This part largely follows \cite{FKT}. Since the difference of the Fermi curves considered in \cite{FKT} and the ones considered in this paper is also a zero measure set, many results stated in Sections $XX$ and $XXI$ of \cite{FKT} can be adapted to the current paper. So we mainly present the results without proof. The interested readers are invited to consult \cite{FKT} for details.
\begin{definition}
Let $\Sigma_\cF$ be a sectorization of $\cF$, in which each sector $s_{\Lambda,l}$ is a rectangle of length $l$ and width $\Lambda$, such that $0\le\Lambda\le l$, $\Lambda\ge l^2$. Let $\bp\in \cF$ and $\Gamma$ be a an $l$-separated subset of $F$. Define
\bea
Mom_{2n-1}(\Gamma,\bp)&=&\{(\bk_1,\cdots,\bk_n)\in\Gamma^{2n-1}\vert\exists\ x_i\in s_{\Lambda,l}(\bk_i), i=1,\cdots,2n-1,\\ && {\rm such\ that}\ x_1+\cdots x_{2n-1}\in s_{\Lambda,l}(\bp) \}.\nn
\eea
\end{definition}
\bd
The tuple $(s_1,\cdots,s_n)\in \Sigma_\cF^{\otimes n}$ is called a configuration of sectors. A configuration of sectors is said to be consistent with the conservation of momentum if the tuple of
vectors $(\bk_1,\cdots,\bk_n)$, with $\bk_i\in s_i$, for $i=1,\cdots,n$, satisfies $\sum_{i=1}^n\bk_i=0$.
\ed
Following exactly the same procedures as in \cite{FKT} and the same techniques employed in proof of Lemma $XX.9$, Proposition $XX.10$, we can prove following proposition: 
\begin{proposition}\label{pro2a}
Let $\cF$ be a $C^2$-differentiable, strictly convex planar Fermi curve. Let $n\ge2$, $\delta\ge l$ and let $I_1,\cdots, I_{2n-1}$ be intervals of length $\delta$ in $F$. Assume that
\be
\frac{1}{3}\omega=\max_{1\le i\neq j\le 2n-1}\min (dist(I_i, I_j), dist(I_i, a(I_j)))>\max(\delta, 4l).
\ee
There exists a constant $K$, which depends on the geometry of the Fermi curve but is independent of the size of sectors, such that for all $l$-separated subsets $\Gamma$ of $F$, all $\bp\in F$,
\be
\# Mom_{2n-1}(\Gamma,\bp)\cap(I_1\times\cdots\times I_{2n-1})\le K n^2\Big(\frac{\delta}{l}+1\Big)^{2n-3}\Big(1+\frac{\Lambda}{l\omega}\Big).
\ee
\end{proposition}

Remark the numerical factor $1/3$ is inessential and can be replaced by any other fractional number between $0$ and $1$ but not very close to $0$ or $1$.
\begin{example}
As an example, consider an anisotropic sectorization $\Sigma^{(j)}_\cF$ of $\cF$ for the single scale $j\ge2$ (cf.eg.\cite{FMRT}), such that each sector $s_{\Lambda,l}$ is a rectangle of length $\g^{-j}$ and width $M^{-j/2}$, in which $M\ge10$ is a fixed constant. Then we have $\delta\sim O(1)M^{-j/2}$, for some order $1$ constant $O(1)$. The centers of the sectors form an $l$-separated set with $l=M^{-j/2}$. We have $\delta/l\sim O(1)$, $\Lambda/l\omega\sim O(1)$, and
\be
\# Mom_{2n-1}(\Gamma,\bp)\cap(I_1\times\cdots\times I_{2n-1})\le const\ n^2O(1)^{2n-3},
\ee
which is bounded for any $n$.
\end{example}
The sector counting lemma is simply a reformulation of the above proposition:
\begin{theorem}[The sector counting lemma]\label{pro2}
Let $\cF$ be a quasi-asymmetric or quasi-symmetric Fermi curve which is $C^2$ differentiable.
Let $I_1,\cdots I_{2n}$ be intervals on the Fermi curve, each of which has length $M^{-j}\le \delta\le M^{-j/2}$, $j\in\ZZZ_+$, $j\ge2$. Let $\bk_i\in\RRR^2$ and $\bk'_i$ be the corresponding projection in $\cF$. Let $K$ be numerical constants which depends on the band structure. Let $s_1$ be a fixed sector. Let $\SS_{2n-1}$, $n\ge2$, be a set of $(2n-1)$-tuples of sectors $\{s_2,\cdots, s_{2n}\}$ such that there exist $\bk_i\in\RRR^2$, $i=1,\cdots 2n$ satisfying 
\be
\bk'_i\in s_i\cap I_i,\quad\vert\bk_i-\bk'_i\vert\le K M^{-j},\quad i=1,\cdots,2n 
\ee
and
\be
\vert\bk_1+\cdots+\bk_{2n}\vert\le K M^{-j}.
\ee
Then the cardinality of the set $\SS_{2n-1}$, noted by $\#\SS_{2n-1}$, is bounded by
\be K^{2n}\Big(\frac{\delta}{M^{-j/2}}\Big)^{2n-3}.\ee
\end{theorem}
Now we consider the sector counting problem with two scales.
\begin{definition}
Let $j>i\ge2$ be two scaling indices. Let $\frac{1}{M^{j-\frac{3}{2}}}\le l\le \frac{1}{M^{(j-1)/2}}$ and 
$\frac{1}{M^{i-\frac{3}{2}}}\le l'\le \frac{1}{M^{(i-1)/2}}$. Let $\Sigma^{(j)}_\cF$ and $\Sigma^{(i)}_\cF$ be two sectorizations of length $l$ at scale $j$ and length $l'$ at scale $i$,
respectively. Define $\#{\rm Cons}(s^{(j)}_1,\cdots, s^{(j)}_m;s^{(i)}_{m+1},\cdots,s^{(i)}_{n})$, in which $s^{(j)}_p$, $p=1,\cdots,m$, are sectors in $\Sigma^{(j)}_\cF$ and $s^{(i)}_q$, $q=m+1,\cdots,n$ are sectors in $\Sigma^{(i)}_\cF$, as the set of all sectors $(s^{(j)}_{m+1},\cdots,s^{(j)}_{n})\in{\Sigma^{(j)}_\cF}^{\otimes (n-m)}$, such that $s^{(j)}_i\cap s^{(i)}_i\neq\emptyset$ for $i=m+1,\cdots,n$, and the sectors $(s_1,\cdots,s_n)$ is consistent with conservation of momentum.
\end{definition}

Following \cite{FKT}, Section $XXI$, we can prove the sector counting lemma for changing of scales.
\begin{theorem}\label{pro3}
Let $\cF$ be a $C^2$-differentiable, strictly convex planar Fermi curve. Let $j>i\ge2$, let $\Sigma^{(j)}_\cF$ and $\Sigma^{(i)}_\cF$ be two sectorizations of $\cF$ defined as above, such that
$l<\frac14 l'$. Let $\omega'\ge 4l'$, and let $s_1\in\Sigma$ and $s_2'\cdots, s_n'\in\Sigma'$ such that $dist(s_k', s_l')\ge \omega'$
and $dist(s_k', a(s_l'))\ge \omega'$ for some $2\le k\neq l\neq n$. Then there exists a positive constant $K$, which is independent of the size of the sectors, such that
\be
\# Cons(s_1; s_2'\cdots s_n')\le K\Big(\frac{l'}{l}\Big)^{n-3}\Big(1+\frac{1}{M^{j-1}l\omega'}\Big).
\ee
\end{theorem}

\begin{proof}
The proof of this theorem is technically identical to that of Lemma $XXI.4$ in \cite{FKT}. So we don't repeat it here.
\end{proof}
\begin{example}
Consider two anisotropic sectorizations $\Sigma^{(j)}_\cF$ and $\Sigma^{(i)}_\cF$, $j\ge i\ge2$, of $\cF$ introduced above. The sectors $s^{(j)}\in\Sigma^{(j)}_\cF$ are of length $l=M^{-j/2}$ and width $\Lambda=M^{-j}$, and the sectors $s^{(i)}\in\Sigma^{(i)}_\cF$ are of length $l'=M^{-i/2}$ and width $\Lambda'=M^{-i}$. Then we have $\omega'\sim O(1)M^{-j/2}$ and
\be
\# Cons(s^{(j)}_1; s^{(i)}_2\cdots s^{(i)}_n)\le const\Big(\frac{l'}{l}\Big)^{n-3}\Big(1+\frac{1}{M^{j-1}l\omega'}\Big)\le O(1)' M^{(j-i)/2},
\ee
in which $O(1)'$ is another positive constant of order $1$.
\end{example}
\section{Conclusions and perspectives}
In this paper we proved the parallelogram lemma and sector counting lemma for any $C^2$ differentiable strictly convex Fermi curves, as a generalization of that considered in \cite{FKT} and \cite{FMRT},\cite{BGM1}. We expect that this important result can be used to solve the inversion problem for more general Fermi surfaces.

\medskip
\noindent{\bf Acknowledgments}
The author is very grateful to H. Kn\"orrer and V. Rivasseau for reading the manuscript and useful suggestions, and to G. Benfatto, A. Giuliani and V. Mastropietro for useful discussions. He is also very grateful to the anonymous referees for useful comments.
The author is supported by NNSFC No.12071099.

\end{document}